\documentclass[11pt, letterpaper]{article}

\usepackage{CJKutf8}

\usepackage[whole]{bxcjkjatype}

\usepackage{cite} 
\usepackage[pdftex]{graphicx}
\usepackage{threeparttable}
\usepackage[margin=1in]{geometry}
\usepackage{amsmath,amssymb,amsfonts,setspace} 
\usepackage{amsthm} 
\usepackage{latexsym}
\usepackage{comment}
\usepackage{bm}
\usepackage{enumerate}
\usepackage{bookmark} 
\usepackage{type1cm} 
\usepackage{color} 
\usepackage{algorithm}
\usepackage[noend]{algorithmic}

\theoremstyle{plain}
\newtheorem{theorem}{Theorem}
\newtheorem{lemma}{Lemma}
\newtheorem{corollary}{Corollary}
\newtheorem{definition}{Definition}
\newtheorem{proposition}{Proposition}
\newcommand{\Indi}{\mathbf{I}}

\newcommand{\MV}{\mathsf{MV}}
\newcommand{\PART}{\ensuremath{\mathsf{PART}}}

\newcommand{\SAPath}{\mathsf{A}}
\newcommand{\LAPath}{\mathsf{B}}

\newcommand{\parent}{\ensuremath{\mathsf{par}}}
\newcommand{\edgep}{\ensuremath{\mathsf{ep}}}
\newcommand{\cyclet}{\ensuremath{\mathsf{cyc}}}
\newcommand{\lca}{\ensuremath{\mathsf{lca}}}
\newcommand{\ests}{\ensuremath{\hat{s}}}

\allowdisplaybreaks[2]

\title{A Subquadratic-Time Distributed Algorithm for Exact Maximum Matching\\\vspace{3mm}}

\author{
Naoki Kitamura\footnote{Nagoya Institute of Technology, Nagoya, Aichi, Japan. E-mail: ktmr522@yahoo.co.jp.}
\and Taisuke Izumi\footnote{Osaka University, Suita, Osaka, Japan. E-mail: t-izumi@ist.osaka-u.ac.jp.}
}
\date{}

\begin{document}

\maketitle


\begin{abstract}
For a graph $G=(V,E)$, finding a set of disjoint edges that do not share any vertices is called a matching problem,
and finding the maximum matching is a fundamental problem in the theory of distributed graph algorithms.
Although local algorithms for the approximate maximum matching problem have been widely studied, exact algorithms has not been much studied. 
In fact, no exact maximum matching algorithm that is faster than the trivial upper bound of $O(n^{2})$ rounds is known for the general instance.
In this paper, we propose a randomized $O(s_{\max}^{3/2}+\log n)$-round algorithm in the CONGEST model, where $s_{\max}$ is the size of maximum matching.
This is the first exact maximum matching algorithm in $o(n^2)$ rounds for general instances in the CONGEST model.
The key technical ingredient of our result is a distributed
algorithms of finding an augmenting path in $O(s_{\max})$ rounds, which
is based on a novel technique of constructing a \emph{sparse certificate} of augmenting paths, which is a subgraph of the input graph preserving at least one augmenting path. To establish a highly parallel construction of sparse certificates, we also propose a new characterization of sparse certificates, which might also be of independent interest.
\end{abstract}

\section{Introduction}
\label{sec:introduction}
\subsection{Background and Our Result}
\label{subsec:background}
A fundamental graph problem is the \emph{maximum (uweighted) matching} problem of finding the maximum cardinality subset of edges not sharing endpoints. In this study, we address the problem of computing exact maximum matchings in a distributed setting, namely, the \emph{CONGEST model}. The CONGEST model is a standard computational model for distributed graph algorithms, where
the network is modeled as an undirected graph $G = (V, E)$ of $n$ nodes and $m$ edges. Each node executes the deployed algorithm following round-based synchrony, and each link
can transfer a small message of $O(\log n)$ bits per round. However, the limited 
bandwidth in the CONGEST model precludes a trivial universal solution for every graph problem, where 
the leader node collects all the topological information of $G$ and solves the problem using a centralized algorithm. This approach takes $O(n^2)$ rounds in the worst case of $m = \Omega(n^2)$. The technical challenge in designing CONGEST algorithms concerns how each node computes
a fragment of the solution without information on the whole input instance. The recent development of design techniques for CONGEST algorithms has yielded many efficient solutions for various graph problems such as the minimum spanning tree\cite{SD98,KKOI2019,GH16,jurdzinski2018mst,GL18-2,HIZ16-2}, distance problems including shortest-path computation\cite{H18,LP13-2,HW12,Nanongkai14,BN18,GL18,FN18}, and flow and cut\cite{GF15,DEMN20,ghaffari2013distributed,nanongkai2014almost,daga2019distributed}. 
Owing to the existence of 
the $O(n^2)$-round universal algorithm, the weakest non-trivial challenge in the design of a CONGEST algorithms is to achieve a subquadradic $o(n^2)$-round upper bound. In contrast
to the universal upper bound, all the problems listed above belong to the class of \emph{global} problems exhibiting an $\Omega(D)$-round lower bound, where $D$ is the diameter of the input graph $G$. Thus, the tight round 
complexities of global problems lie between $\Theta(n^2)$ and $\Theta(D)$. For many of global problems, near-tight complexity bounds, typically $\tilde{\Theta}(\sqrt{n} + D)$ rounds or $\tilde{\Theta}(n)$ rounds, have been proved\cite{FHW12,bacrach2019hardness,lowerbound}. 

Many studies in the context of approximation algorithms provide insight into the complexity of the maximum matching problem. Table~\ref{tab:MM} lists the known algorithms, where $s_{\max}$ is defined the cardinality of the maximum matching. While $O(1)$ approximation admits local solutions (i.e., $o(D)$-round algorithms), 
the complexity of the exact maximum matching problem makes it expensive. Precisely, 
following the lower bound of Ben-Basat et al.~\cite{BKS18}, there exists an instance of diameter $\Omega(n)$ and maximum matching size $\Omega(n)$ that exhibits an $\Omega(n)$-round lower bound. This lower bound was originally proved in the \emph{LOCAL} model, but it trivially holds in the CONGEST model as well. Therefore, the exact maximum matching 
problem is placed in the class of global problems. Parametrizing the complexity by both $n$ and $D$, it is possible to obtain the non-trivial lower bound of $\Omega(D + \sqrt{n})$ rounds for the exact computation of the maximum matching\footnote{This lower bound was not explicitly shown in previous literatures, but it is derived using the lower-bound graph almost same as that used in the lower-bound proof for the fractional maximum matching by Ahmadi et~al.\cite{AKO18}}. However, the corresponding upper bound is yet to be found. For the exact maximum matching problem in general graphs, no known algorithm achieves non-trivial $O(n^2)$ rounds. In addition, Bacrach et al.~\cite{bacrach2019hardness} pointed out that
the bound of $\Omega(\sqrt{n} + D)$ rounds is a strong barrier because the standard framework of two-party communication complexity is unlikely to deduce any improved lower bound. These observations demonstrate the difficulty of revealing the inherent complexity of the exact maximum matching in the CONGEST model.

The objective of this paper is to shed light on the complexity gap of the exact maximum matching problem in the CONGEST model. We present the main theorem of this paper in the CONGEST model below. 
\begin{theorem}
\label{theo:mm}
For any input graph $G$, there exists a randomized CONGEST algorithm to compute the maximum matching that terminates within $O\left(s_{\max}^{3/2}+\log n \right)$ rounds with probability $1 - 1/n^{\Theta(1)}$.
\end{theorem}
To the best of our knowledge, the proposed algorithm is the first to compute the exact maximum matching algorithm in $o(n^2)$ rounds for general input instances in the CONGEST model.

\subsection{Technical Outline}
Our algorithm follows the standard technique of finding \emph{augmenting paths}. If an augmenting path is found, the current matching is improved by flipping the labels of matching edges and non-matching edges along the path. It is well known that the current matching is the maximum if and only if there exists no augmenting path in $G$ with respect to the current matching. Hence, the maximum matching problem is reduced to the task of 
finding augmenting paths $s_{\max}$ times. In the CONGEST model, this approach faces difficulty in the situation where any augmenting path with respect to the current 
matching is long (i.e., consisting of $\Theta(n)$ edges). It should be emphasized that BFS-like approaches do not work for finding augmenting paths in general graphs because the shortest alternating walk is not necessarily simple because of the existence of odd cycles. Thus, it is not trivial to even compute an augmenting path with a running time linearly 
dependent on its length. The key ingredient of our approach is two new algorithms for finding augmenting paths. They run in $O(\ell^2)$ rounds and $O(s_{\max})$ rounds respectively, where $\ell$ is the length of the shortest augmenting path for the current matching. Roughly, our algorithm switches between these two algorithms according to the current matching size. The running-time bound is obtained using the following seminal observation by 
Hopcroft and Karp:
\begin{proposition}[Hopcroft and Karp~\cite{HK73}]
\label{pro:hk}
Given a matching $M \subseteq E$ of a graph $G$, there always exists an augmenting path of length less 
than $\lfloor2s_{\max}/k\rfloor$ if the current matching size is at most the maximum matching size $s_{\max}$ minus $k$.
\end{proposition}
Our augmenting path algorithms utilize Ahmadi's verification algorithm of maximum matching~\cite{AK20}, in which each node returns the length of the shortest odd/even alternating paths from 
a given source (unmatched) node. The construction of the $O(\ell^2)$-round algorithm is relatively straightforward. It is obtained by iteratively finding the predecessor of each node in an augmenting path by sequential $O(\ell)$ invocations of the verification algorithm. The technical highlight of the proposed algorithm is 
the design of the $O(s_{\max})$-round algorithm. The $O(s_{\max})$-round algorithm constructs a 
\emph{sparse certificate}, which is a sparse (i.e., containing $O(s_{\max})$ edges) subgraph of $G$ preserving the reachability between two nodes by alternating paths. That is, a sparse certificate contains an augmenting path if and only if the 
original graph admits an augmenting path. By the sparseness property, a node can collect all the information on the sparse certificate within $O(s_{\max})$ rounds, trivially allowing the centralized solution of finding augmenting paths. To establish a highly parallel construction of sparse certificates, we also propose a new characterization of
sparse certificates, which might also be of independent interest.
\subsection{Related Works}
In the LOCAL model, it is known that no $o(1/\epsilon)$ algorithm exists for the $(1-\epsilon)$-approximate maximum matching problem~\cite{BKS18}.
Together with the $\Omega(\sqrt{\log n}/\log \log n)$-round lower bound reported by Kuhn et~al.~\cite{KMW16}, the lower bound in the LOCAL model is obtained as $\Omega(1/\epsilon+\sqrt{\log n/\log \log n})=((\log n)/\epsilon)^{\Omega(1)}$.
Ghaffari et~al.~\cite{GKM17} showed a $((\log n)/\epsilon)^{O(1)}$ upper bound for the $(1-\epsilon)$ approximate maximum matching problem. 
By combining these results, we infer that the time complexity of solving the $(1-\epsilon)$ approximate maximum matching problem is  $(\log n/\epsilon)^{\Theta(1)}$ in the LOCAL model.
Ben-Basat et~al. also proved the lower bound as $\Omega(|M|)$ in the LOCAL model~\cite{BKS18}.

Many literatures have addressed the maximum matching problem in the CONGEST model (see Table~\ref{tab:MM}).
Loker et al.~\cite{LPP08} presented the first approximation algorithm in the CONGEST model,
which is a randomized algorithm to compute $(1-\epsilon)$-approximate maximum matching in $O(\log n)$ rounds for any constant $\epsilon>0$.
The running time of the algorithm depends exponentially on $1/\epsilon$.
Bar Yehuda et~al.~\cite{BCGS17} improved the algorithm and proposed an $O(\log \Delta/\log \log \Delta)$-round algorithm of computing $(1-\epsilon)$-approximate matching for any constant $\epsilon>0$, where $\Delta$ is maximum degree of the graph.
Fabin et~al.~\cite{FTR06} has shown of $\Omega(\log \Delta/\log \log \Delta)$ rounds if $\log \Delta\leq \sqrt{\log n}$ holds. Ben-Basat et~al.~\cite{BKS18} proposed a deterministic $\tilde{O}(s_{\max}^2)$-round CONGEST algorithm. They also proposed a $(1/2-\epsilon)$ approximate algorithm in $\tilde{O}(s_{\max}+(s_{\max}/\epsilon)^2)$ rounds.
Ahmadi et~al.~\cite{AKO18} proposed a deterministic $(2/3-\epsilon)$ approximate maximum matching algorithm in general graphs, which runs in $O(\log \Delta/\epsilon^2 + (\log^2 \Delta+\log^{*}n)/\epsilon)$ rounds.
They also presented an $\tilde{O}(M)$-round algorithm and $O((\log^2 \Delta+\log^{*}n)/\epsilon)$-round $(1-\epsilon)$ approximate algorithm in bipartite graphs.
However, no $o(n^{2})$-round algorithm for solving the exact maximum matching problem in the CONGEST model has been proposed so far.

In addition to distributed computing, many studies have considered centralized exact maximum matching algorithms. 
Edmonds presented the first centralized polynomial-time algorithm for the maximum matching problem~\cite{Edmonds,Edmonds2} by following the seminal blossom argument. Hopcroft and Karp proposed a phase-based algorithm of finding multiple augmenting paths~\cite{HK73}. Their algorithm finds a maximal set of pairwise disjoint shortest augmenting paths in each phase. They showed that $O(\sqrt{n})$ phases suffice to compute the maximum matching and proposed an algorithm of implementing one phase in $O(m)$ time for bipartite graphs. Several studies have reported phase-based algorithms for general graphs that attain $O(\sqrt{n}m)$ time~\cite{Blum1990,Vazirani2020,GT91}.

\begin{table*}[tb]
\caption{Lower and upper bounds of the maximum matching in the CONGEST model}
\center
\label{tab:MM}
\begin{tabular}{| l c c c|}
\hline
Algorithm & Time Complexity &Approximation Level  & Remark\\
\hline
Ben-Basat et al.~\cite{BKS18} & $\Omega(|s_{\max}|)$ & exact &LOCAL\\
Fabin et~al.~\cite{FTR06}& $\Omega\left(\frac{\log \Delta}{\log \log \Delta} \right)$& constant $\epsilon$ & $\log \Delta\leq \sqrt{\log n}$\\
Ben-Basat et~al.~\cite{BKS18}& $\Omega\left(\frac{1}{\epsilon}\right)$ & $1-\epsilon$ &LOCAL\\
Kuhn et~al.~\cite{KMW16}& $\Omega\left(\sqrt{\frac{\log n}{\log \log n}}\right)$ & $1-\epsilon$ &LOCAL\\
Ben-Basat et~al.~\cite{BKS18} & $\tilde{O}(s_{\max}^2)$ &exact & \\
Ahmadi et~al.~\cite{AKO18}& $\tilde{O}\left(s_{\max}\right)$& exact & bipartite\\
Bar-Yehuda et~al.~\cite{BCGS17}& $O\left(\frac{\log \Delta}{\log \log \Delta}\right)$ &constant $\epsilon$& \\
Lotker et~al.~\cite{LPP08}& $O\left(\frac{2^{\frac{2}{\epsilon^2}}\log s_{\max}\log n}{\epsilon^4}\right)$ & $1-\epsilon$ & \\
Ahmadi et~al.~\cite{AKO18}& $O\left(\frac{\log^2 \Delta+\log^{*}n}{\epsilon}\right)$& $1-\epsilon$ & bipartite\\
Ben-Basat et~al.~\cite{BKS18} & $\tilde{O}\left(s_{\max}+\left(\frac{s_{\max}}{\epsilon}\right)^2\right)$ & $\frac{1}{2}-\epsilon$ &\\
Ahmadi et~al.~\cite{AKO18}& $O\left(\frac{\log \Delta}{\epsilon^2}+\frac{\log^2 \Delta+\log^{*}n}{\epsilon}\right)$& $\frac{2}{3}-\epsilon$ &\\
\textbf{Our result}& \boldmath{$\tilde{O}\left(s_{\max}^{3/2}\right)$} & \textbf{exact} &\\
\hline
\end{tabular}
\end{table*}
\section{Preliminaries}
\label{sec:pre}
\subsection{CONGEST Model}
The vertex set and edge set of a given graph $G$ are, respectively, denoted by $V(G)$ and 
$E(G)$.
A distributed system is represented by a simple undirected connected graph $G = (V(G), E(G))$.
Let $n$ and $m$ be the numbers of nodes and edges, respectively. 
The diameter of a given subgraph $H \subseteq G$ is denoted by $D(H)$.
Nodes and edges are uniquely identified by integer values, 
which are represented by $O(\log n)$ bits. The set of edges incident to $v \in V(G)$ 
is denoted by $I_G(v)$. In the CONGEST model, the computation 
follows round-based synchrony. In one round,
each node $v$ sends and receives $O(\log n)$-bit messages through the edges in $I_G(v)$ 
and executes local computation following its internal state, local random bits, 
and received messages. It is guaranteed that every message sent in a round is delivered to the destination within the same round. 
Each node has no prior knowledge of the network topology, except for its neighborhood IDs. 
We use the labeling of nodes and/or edges for specifying inputs and outputs of algorithms. 
Each node has information on the label(s) assigned to itself and those assigned to its incident edges.
A \emph{walk} $W$ of $G$ is an alternating sequence $W = v_0, e_1, v_1, e_2, \dots, e_\ell, 
v_\ell$ of vertices and edges such that $e_i = (v_{i-1}, v_i)$ holds for any 
$1 \leq i \leq \ell$. A walk $W$ is often treated as a subgraph of $G$. A walk 
$W = v_0, e_1, v_1, e_2, \dots, e_\ell, v_\ell$ is called a (simple) \emph{path} if 
every vertex in $W$ is distinct. For any walk $W = v_0, e_1, v_1, \dots, v_\ell$ of $G$, 
we define $W \circ u$ as the walk obtained by adding $u$, satisfying $(v_\ell, u) \in 
E(G)$, to the tail of $W$. For any edge $e = (v_\ell, u)$, we also define  
$W \circ e  = W \circ u$. Given a walk $W$ containing a node $u$, we 
denote by $W^p_u$ and $W^s_u$ the prefix of $W$ up to $u$ and the suffix of $W$ from $u$, respectively. We also denote the inversion of the walk $W = v_0, e_1, v_1, \dots, v_\ell$ 
(i.e., the walk $v_\ell, e_\ell, v_{\ell - 1}, e_{\ell -1}, \dots, v_0$) by $\overline{W}$. The length of a walk $P$ is represented by $|P|$.
\subsection{Matching and Augmenting Path}
For a graph $G=(V,E)$, a matching $M\subseteq E$ is a set of edges that do not share endpoints. A node $v$ is called a \emph{matched node} if $M$ intersects $I_G(v)$, 
or an \emph{unmatched node} otherwise.
A path $P = v_0, e_0, v_1, e_1, \dots, v_\ell$ is called an \emph{alternating path} 
if $\Indi_M(e_i) + \Indi_M(e_{i+1}) = 1$ holds for any $1 \leq i \leq 
\ell-1$\footnote{The indicator function $\Indi_X(x)$ returns one if $x \in X$ and 
zero otherwise.}. If the length $|P|$ of $P$ satisfies $|P| \mod 2 = \theta$,
$P$ is called \emph{$\theta$-alternating}. The value $\theta$ is called the 
\emph{parity} of $P$. By definition, any $0$-alternating ($1$-alternating) path 
from an unmatched node $f$ finishes with a non-matching (matching) edge. Oue to a technical 
issue, we regard the path of length zero as a $0$-alternating path. For any 
$\theta \in \{0, 1\}$ and $u, v \in V(G)$, we define
$r^{\theta}(u, v)$ as the length of the shortest $\theta$-alternating path between  
$u$ and $v$. 
An \emph{augmenting path} is an alternating path connecting two unmatched nodes. 
We say that $(G, M)$ has an augmenting path if there exists an augmenting path in $G$ 
with respect to $M$. The following proposition is a well-known fact in the maximum 
matching problem.
\begin{proposition}
Given a matching $M \subseteq E(G)$ of graph $G$, $M$ is the maximum matching if and only 
if $(G, M)$ has no augmenting path.
\end{proposition}
\subsection{Approximate Maximum Matching}
Our algorithm uses an $O(1)$-approximate upper bound for the maximum matching size of the input graph. To obtain
the upper bound, we run the $O(\log n)$-round randomized maximal matching algorithm~\cite{II86} as a preprocessing step.
Let $M^{\ast}$ be the computed maximal matching. Since 
any maximal matching is a $(1/2)$-approximate maximum matching, 
one can obtain the bound $2|M^{\ast}| \geq s_{\max}$. The size 
$s_{\max}$ is at least half of the diameter $D(G)$, and thus we can spend $O(D(G)) = O(s_{\max})$ rounds for counting and propagating the number of edges in $M^{\ast}$.
That is, it is possible to provide each node with the value of $2|M^{\ast}|$ by 
the preprocessing of $O(D(G)+\log n)=O(s_{\max}+\log n)$ rounds. In the following argument, we denote $\hat{s} = 2s^{\ast}$, the value of which is available to each 
node.
\subsection{Maximum-Matching Verification Algorithm}
Our algorithm uses the algorithm by Ahmadi et~al.'s~\cite{AK20} for maximum-matching verification as a building block. Although the original algorithm is designed
for the verification of maximum matching, it provides each node with information on the length of alternating paths to the closest unmatched nodes. Precisely, the following lemma holds.
\begin{theorem}[Ahmadi et al.~\cite{AK20}]
\label{theo:verification}
Assume that a graph $G = (V, E)$ and a matching $M \subseteq E$ are given, and let 
$W$ be the set of all unmatched nodes. There exist two $O(\ell)$-round randomized CONGEST algorithms $\MV(M, \ell, f)$ and $\PART(M, \ell)$ that output the following information at every node $v \in V(G)$
with a probability of at least $1-1/n^c$ for an arbitrarily large constant $c > 1$.
\begin{enumerate}
\item Given $M$, a nonnegative integer $\ell$, and a node $f \in W$, 
$\MV(M, \ell, f)$ outputs the pair $(\theta, r^{\theta}(f, v))$ at each node $v$
if $r^{\theta}(f, v)\leq \ell$ holds (if the condition is satisfied for both $\theta = 0$ and $\theta = 1$, $v$ outputs two pairs). The algorithm $\MV(M, \ell, f)$ is initiated only by the node $f$ (with the value $\ell$), and other nodes do not require information on the ID of $f$ and value $\ell$ at the initial stage.
\item The algorithm $\PART(M, \ell)$ outputs a partition $V^1, V^2, \dots, V^{N}$ of $V(G)$ (as the label $i$ for each node in $V^i$) such that (a) the subgraph $G^i$ induced by
$V^i$ contains exactly two unmatched nodes $f^i$ and $g^i$ as well as an augmenting path between $f^i_1$ and $g^i_2$ of length at most $\ell$ and (b) the diameter of $G^{i}$ is $O(\ell)$. 
\end{enumerate}
\end{theorem}
While the original paper~\cite{AK20} presents a single algorithm returning the outputs of both $\MV$ and $\PART$, we intentionally separate it into two algorithms with different roles for clarity. Note that our matching-construction algorithm uses random bits only in the runs of these algorithms. As our algorithm activates them in $O(\mathrm{poly}(n))$ time as subroutines, we can guarantee that
our algorithm has a high probability of success by taking a sufficiently large $c$. Hence, we do not pay much 
attention to the failure probability of our algorithm. Any stochastic
statement in the following argument also holds with probability $1 - n^{c}$
for an arbitrary constant $c > 1$.
\section{Computing the Maximum Matching in CONGEST}
As explained in the introduction, the maximum matching problem is reducible to the problem of finding an augmenting path.
We first present two key results below.
\begin{lemma}
\label{lma:square}
Let $M$ be a matching of $G$. Provided that
$(G, M)$ has exactly two unmatched nodes $f, g \in V_G$ and contains 
an augmenting path of length at most $\ell$ between $f$ and $g$, 
there exists an $O(\ell^2)$-round randomized algorithm that outputs 
an augmenting path connecting $f$ and $g$.
\end{lemma}
\begin{lemma}
\label{lma:n}
Let $M$ be a matching of $G$. Provided that
$(G, M)$ has exactly two unmatched nodes $f, g \in V_G$ and contains 
an augmenting path between $f$ and $g$, 
 there exists an $O(n)$-round randomized algorithm that outputs an augmenting path that includes $f$.
\end{lemma}
The outputs of both algorithms are the labels to the edges in the computed 
augmented path. To prove the lemmas, one can utilize the output of the algorithm $\PART$. We first run the verification algorithm 
$\PART(M, \ell)$ (for Lemma~\ref{lma:square}) or $\PART(M, \hat{s})$ (for Lemma~\ref{lma:n}) as a preprocessing step and then execute the algorithms 
of Lemma~\ref{lma:square} or \ref{lma:n} for each $G^{i}$ output by 
$\PART$ independently. Note that each $G^{i}$ contains only matched nodes 
and two unmatched nodes; thus, $|V(G^{i})| \leq 2|M|+2$ holds for any $G^i$. Then, the following corollary is deduced:
\begin{corollary}
There exist two randomized algorithms $\SAPath(M, \ell)$ and $\LAPath(M)$ satisfying the following conditions, respectively:
\begin{itemize}
    \item For any graph $G = (V, E)$ and matching $M \subseteq E$, 
    $\SAPath(M, \ell)$ finds a nonempty set of vertex-disjoint augmenting paths within $O(\ell^2)$ rounds if $(G, M)$ has an augmenting path of length at most $\ell$.
    \item For any graph $G = (V, E)$ and matching $M \subseteq E$, $\LAPath(M)$ finds
    a nonempty set of vertex-disjoint augmenting paths of $(G, M)$ 
    within $O(|M|)$ rounds if $(G, M)$ has an augmenting path.
\end{itemize}
\end{corollary}
We present an $O\left(s_{\max}^{3/2}+\log n\right)$-round algorithm for computing the maximum matching
using the algorithms $\SAPath(M, \ell)$ and $\LAPath(M)$.
The pseudocode of the whole algorithm 
is presented in Algorithm~\ref{alg:mm}. 
It basically follows the standard idea of centralized maximum matching algorithms, i.e., finding an augmenting path and improving the current matching iteratively. 
The first $\ests-\sqrt{\ests}$ iterations use $\SAPath(M,\ell)$ (lines 1--4), and the remaining $\sqrt{\hat{s}}$ iterations use $\LAPath(M)$.
In the $i$-th iteration, the algorithm $\SAPath(M, \ell)$ runs with $\ell=\lceil2\ests/(2\ests-i)\rceil$. This setting comes from Proposition~\ref{pro:hk}.
The improvement of the current matching by a given augmenting path is simply a local operation and is realized by flipping the labels of matching edges and 
non-matching edges on the path. The correctness and running time of Algorithm~\ref{alg:mm} are analyzed below.
\begin{algorithm*}[tb]
\caption{Constructing a maximum matching in $O(n^{3/2})$ rounds.}
\label{alg:mm}
\begin{algorithmic}[1]
\FOR{$i=1;i\leq \ests-\sqrt{\ests};i++$}
\STATE run the algorithm $\SAPath(M,\ell)$ with $\ell = \lceil2\ests/(\ests-i)\rceil$ for $O(\ell)$ rounds.
\IF{$\SAPath(M,\ell)$ finds a nonempty set of vertex-disjoint augmenting paths within $O(\ell)$ rounds,}
\STATE improve the current matching using the set of vertex-disjoint augmenting paths.
\ENDIF
\ENDFOR
\FOR{$i=1;i\leq \sqrt{\ests};i++$}
\STATE run the algorithm $\LAPath(M)$ for $O(\ests)$ rounds.
\IF{$\LAPath(M)$ finds a nonempty set of vertex-disjoint augmenting paths within $O(\ests)$ rounds,}
\STATE improve the current matching $M$ using the set of vertex-disjoint augmenting paths.
\ENDIF
\ENDFOR
\end{algorithmic}
\end{algorithm*}
\begin{lemma}
Algorithm~\ref{alg:mm} constructs a maximum matching with high probability in $O\left(s_{\max}^{3/2}+\log n\right)$ rounds.
\end{lemma}
\begin{proof}
Let $s(i)$ be the matching size at the end of $i$ iterations of the algorithm $\SAPath(M,\ell)$.
We show that $s(\ests-s_{\max}+j)\geq j$ holds for any $0 \leq j \leq s_{\max}-\sqrt{\ests}$.
It implies that the matching size is at least $s_{\max}-\sqrt{\hat{s}}$ after the application of $\SAPath(M,\cdot)$. 
Therefore, the maximum matching is constructed by $\sqrt{\ests}$ iterations of the algorithm $\LAPath(M)$.
The proof of the statement above follows the induction on $j$.
(Basis) If $j = 0$, the statement trivially holds.
(Inductive step) As the induction hypothesis, suppose $s(\ests-s_{\max}+j')\geq j'$ holds.
If $s(\ests-s_{\max}+j')\geq j'+1$, then the statement holds.
Therefore, we consider the case in which $s(\ests-s+j')=j'$ holds.
By Proposition~\ref{pro:hk}, there exists an augmenting path of length at most $\lfloor2s_{\max}/(s_{\max}-j')\rfloor\leq 2\ests/(s_{\max}-j')= 2\ests/(\ests-(\ests-s_{\max}+j'))\leq 2\ests/(\ests-(\ests-s_{\max}+(j'+1)))$ at the end of  $\ests-s_{\max}+j'$ iterations of the algorithm $\SAPath(M,\ell)$.
Hence, the size of the matching is increased by at least one in the ($\ests-s_{\max}+j'+1$)-th iteration.

Now, we show the running-time analysis of Algorithm~\ref{alg:mm}.
As $\SAPath(M,\ell)$ is repeated $\ests-\sqrt{\ests}$ times and $\LAPath(M)$ is repeated $\sqrt{\ests}$ times, the running time of Algorithm~\ref{alg:mm} is as follows.
\begin{align*}
&O\left(s_{\max}+\log n\right)+ O\left(\sum_{i=1}^{\ests-\sqrt{\ests}}\left(\left\lceil \frac{2\ests}{\ests-i} \right\rceil\right)^2\right)+O\left(\ests\sqrt{\ests}\right)\\
=&O\left(\sum_{i=1}^{\ests-\sqrt{\ests}}\left( \frac{\ests}{\ests-i} \right)^2+\ests-\sqrt{\ests}+\ests\sqrt{\ests} +\log n\right)\\
=&O\left(\sum_{i=\sqrt{\ests}}^{\ests-1}\left( \frac{\ests}{i} \right)^2+\ests\sqrt{\ests} +\log n\right)\\
=&O\left(\ests^{2}\sum_{i=\sqrt{\ests}}^{\ests-1}\left( \frac{1}{i} \right)^2+\ests\sqrt{\ests} +\log n\right)\\
=&O\left(\ests^{2}\frac{1}{\sqrt{\ests}}+\ests\sqrt{\ests} +\log n\right)\\
= &O\left(\ests^{3/2} + \log n\right)\\
= &O\left(s_{\max}^{3/2} + \log n\right).\\
\end{align*}
\end{proof}
The following sections are devoted to proving Lemmas~\ref{lma:square} and \ref{lma:n}. Since the presented algorithms are intended to run in each $G^i$ 
returned by the preprocessing run of $\PART(M, \cdot)$, without loss of generality, we assume 
that $G$ has exactly two unmatched nodes $f$ and $g$ with an augmenting path between them.
In addition, it is assumed that one of $f$ and $g$ is 
elected as a primary unmatched node (referred to as $f$ hereafter). This election process is easily implemented in $O(\ell)$ rounds because the distance between $f$ and $g$ is at most $\ell$.
When we argue the existence of augmenting or alternating paths in a subgraph $H = (V(H), E(H))$ of $G$, the matching $M \cap E(H)$ of graph $H$ is considered without explicit notice.
Given a subgraph $H \subseteq G$, we denote the length of the shortest odd (even) alternating path from $f$ to $v$ in $H$ by $r^{1}_{H}(f,v)$ ($r^{0}_{H}(f,v)$). If no odd or even alternating path exists from $f$ to $v$ in $H$, then we define $r^{1}_{H}(f,v) = \infty$ or $r^{0}_{H}(f,v) = \infty$.
As sentinels, we also define $r^{0}_{H}(f,f)$ as $\infty$ and $r^{1}_{H}(f,f)$ as $0$.
\section{Construction of Augmenting Path in $O(\ell^2)$ Rounds}
\label{sec:square}

\subsection{Outline}
Let $P = v_0, v_1, \dots, v_\ell$ be the shortest augmenting path 
from $f$ to $g$ (i.e., $f = v_0$ and $g = v_\ell$) and $P_{i}=P^{s}_{v_i}$ for short. The key idea of the algorithm is to find the predecessor of each node $v_i$ along $P$ sequentially. Note that it does not suffice to choose
a neighbor $v$ of $v_i$ with $r^{\theta}_G(f,v) = i - 1$ and $\Indi_{M}(v_{i},v)=\theta$ for $\theta = (i - 1) \mod 2$ as the predecessor. This strategy is problematic in the scenario in which there exists two neighbors $v$ and 
$u$ such that $r^{\theta}_G(f,v) = r^{\theta}_G(f,u) = i - 1$ and $\Indi_{M}(v_{i},u)=\Indi_{M}(v_{i},v)=\theta$ for $\theta = (i - 1) \mod 2$, where $u$ is the correct successor. While $v$ is guaranteed to have the alternating 
path $Q$ from $f$ to $v$ of length $i - 1$, it can intersect $P_i$. Then, the 
concatenation $Q \circ (v_{i},v) \circ P_i$ is not simple. That is, it is not an augmenting path.
To avoid this scenario, the algorithm finds the predecessor of $v_i$ in the graph 
$G - P_{i}$, where $G - P_{i}$ is the induced graph by $V(G)\backslash V(P_i)$. If some neighbor $v$ of $v_i$ satisfies 
$r^{\theta}_{G - P_{i}}(f,v) = i - 1$ and $\Indi_{M}(v_{i},v)=1-\theta$, the concatenated walk $Q \circ (v_{i},v) \circ P_i$ is guaranteed
to be simple.

\begin{algorithm*}[tb]
\caption{Construction of the augmenting path $CAP((G,M),f,g,\ell)$ for node $v_{i}$.}
\label{alg:cap}
\begin{algorithmic}[1]
\REQUIRE The path $P_{0}$ is an augmenting path with length $\ell$ from $f$ to $g$.
\STATE $P_{0},P_{1},\dots,P_{\ell}$: initially $\emptyset$.
\STATE $\mathsf{target}= g$
\FOR{$i=1;i\leq \ell; i++$}
	\IF{$h$ is even}
	\STATE $\mathsf{target}$ chooses the node $v_{\ell - i}$ that satisfies $\Indi_{M}((\mathsf{target},v_{\ell -i }))= 1$.
		\STATE $P_{\ell-i}\leftarrow P_{\ell-i+1}\cup\{(\mathsf{target},v_{\ell- i })\}$.
		\STATE $\mathsf{target} \leftarrow v_{\ell- i}$.
	\ELSE
		\STATE run the algorithm $\MV(M,\ell-i,f)$ with the subgraph $H_{\ell-i+1}$ induced by $V(G - P_{\ell-i+1})$ as the input.
		\STATE  for any $v\in V(G - P_{\ell-i+1})$, the node $v$ sends $r^{0}_{H_{\ell-i+1}}(f,v)$ to its neighborhood.
		\STATE $\mathsf{target}$ chooses the node $v_{\ell-i}$ that satisfies $\Indi_{M}((\mathsf{target},v_{\ell-i}))= 0$ and $r^{0}_{H_{\ell-i+1}}(f,v_{\ell-i})=\ell-i$.
		\STATE $P_{\ell-i} \leftarrow P_{\ell-i+1}\cup \{(\mathsf{target},v_{\ell-i})\}$.
		\STATE $\mathsf{target} \leftarrow v_{\ell-i}$.
	\ENDIF	
\ENDFOR
\end{algorithmic}
\end{algorithm*}

\subsection{Algorithm Details}
Algorithm~\ref{alg:cap} details the algorithm for constructing the augmenting path in $O(\ell^2)$ rounds. The algorithm consists
of $\ell$ steps. In the $i$-th step, it finds the predecessor of $v_{\ell - i}$.
Assume that the algorithm has already found $P_{\ell - i + 1}$ at the beginning of 
the $i$-th step. Any node in $V(P_{\ell-i + 1}) \setminus \{v_{\ell -i +1}\}$ quits the algorithm (with 
the information of the predecessor in $P_i$), and thus, the nodes still running the algorithm are given by $V(G - P_{\ell-i+1})$. If $i$ is even, the edge $(v_{\ell-i}, v_{\ell-i + 1})$ is the 
matching edge, and thus, the algorithm determines the neighbor of $v_{\ell-i+1}$ connected by
the edge with $M$ as the predecessor. Otherwise, the nodes still participating in the algorithm
run $\MV(M, \ell-i+1,f)$ (that is, they run in the graph $G - P_{\ell - i+1}$) The algorithm decides an arbitrary neighbor $v$ of $v_i$ satisfying $r^{0}_{G - P_{\ell - i+1}}(f,v) = \ell - i - 1$ and $\Indi_{M}(v,v_{i})=0$ as the
predecessor of $v_{\ell-i}$. 

\begin{lemma}
\label{lma:square_jus}
Algorithm ~\ref{alg:cap} constructs an augmenting path between $f$ and $g$ with high probability in $O(\ell^2)$ rounds.
\end{lemma}
\begin{proof}
Let $z_{0} = g$ and $z_{i}$ be the node that satisfies $\mathsf{target}=z_{i}$ at the end of the $i$-th iteration for $1\leq i \leq \ell$.
Let $H_{i}$ be a subgraph induced by $V(G-P_{i})$.
We prove the statement that $P_{\ell-h}$ is a $(h\mod 2)$-alternating path between $z_{0}$ and $z_{h}$. As $r^{0}_{G}(f,z_{\ell})\leq r^{0}_{H_{1}}(f,z_{\ell})=0$, $z_{\ell}=v$ holds, and thus, we obtain $P_{0}$ as an augmenting path of length $\ell$ from $f$ to $g$ by setting $h=\ell$. 
The proof follows the induction on $h$.
(Basis) Since $z_{0}$ chooses the node $z_{1}$ that satisfies $\Indi_{M}((\mathsf{target},v_{\ell-1}))= 0$ and $r^{0}_{H_{\ell}}(f,v_{\ell-1})=\ell-1$ in the first iteration of Algorithm~\ref{alg:cap}, $P_{\ell-1}=\{(z_{0},z_{1})\}$ is a $1$-alternating path between $z_{0}$ and $z_{1}$. 
(Inductive Step) As the induction hypothesis, suppose there exists a $(h' \mod 2)$-alternating path between $z_{0}$ and $z_{h'}$ at the end of the $h'$-th iteration. 
Because $r^{h' \mod 2}_{H_{\ell-h'+1}}(f,z_{h'}) = \ell-h'$ holds by the definition of $z_{h'}$, there exists an edge $(z_{h'},v)$ that satisfies $\Indi_{M}((z_{h'},v))= h' \mod 2$, and $r^{(h'+1) \mod 2}_{H_{\ell-h'}}(f,v)=\ell-h'-1$ holds.
Therefore, $z_{h'}$ can choose the node $z_{h'+1}$ that satisfies $\Indi_{M}((\mathsf{target},z_{h'+1}))= h' \mod2$ and $r^{0}_{H_{\ell-h'}}(f,z_{h'+1})=\ell-h'-1$ in the $(h'+1)$-th iteration of Algorithm~\ref{alg:cap}.
Hence, $P^{\ell-h'}\circ \{(z_{h'},z_{h'+1})\}$ is a $((h+1)\mod 2)$-alternating path between $z_{0}$ and $z_{h'+1}$ at the end of the $(h'+1)$-th iteration. 

We show the running-time analysis of Algorithm~\ref{alg:cap}.
The algorithm consists of $\ell$ iterations.
As each iteration is obliviously implemented in $O(\ell)$ rounds, the running time of Algorithm~\ref{alg:cap} is $O(\ell^{2})$ rounds.
\end{proof}
Theorem~\ref{lma:square} trivially follows from Lemma~\ref{lma:square_jus}.

\section{Construction of Augmenting Path in $O(n)$ Rounds}
\label{sec:n}
\subsection{Outline}
\label{sec:out_n}
We first introduce several auxiliary notions and definitions. 
Given a subgraph $H \subseteq G$ and $\theta \in \{0, 1\}$, a node $v \in V_H$ is called \emph{$\theta$-reachable} in $H$ if $r^\theta_H(f,v)$ is finite. In addition, $v$ is called \emph{bireachable} in $H$ if it is both 1-reachable and 0-reachable in $H$.
A node that is neither 1-reachable nor 0-reachable in $H$ is called \emph{unreachable} in $H$.
A node that is $\theta$-reachable for some $\theta \in \{0, 1\}$ in $H$ but not 
bireachable in $H$ is called \emph{strictly $\theta$-reachable} in $H$.
Given two spanning subgraphs $H_1$ and $H_2$ of $G$, 
we say that a node $v \in V(H_1)$ \emph{preserves the reachability} of $H_2$ in $H_1$ if for any $\theta \in \{0, 1\}$, 
the $\theta$-reachability of $v$ in $H_2$ implies that in $H_1$. A graph $H_1$ is said to preserve the reachability of $H_2$ if any node $v \in V(H_1)$ preserves the reachability of $H_2$ in $H_1$, which is denoted by $H_1 \succ H_2$. 
We define $r_{H}(f,v)=\min_{\theta \in \{0, 1\}} r^\theta_H(f,v)$ and $\gamma_H(v) = \mathrm{argmin}_{\theta \in \{0, 1\}} r^\theta_H(f,v)$. Note that $r^0_H(f,v) = r^1_H(f,v)$ does not hold, because $r^0_H(f,v)$ is even and $r^1_H(f,v)$ is odd.  
When $r^0_H(f,v) = \infty$ and  $r^1_H(f,v) = \infty$ hold, $\gamma_H(v)$ is defined as zero. 
We assume that any node $v$ unreachable from $f$ in $G$ does not join our algorithm.
Therefore, without loss of generality, we assume that none of the nodes $v \in V_G$ are unreachable in $G$ without loss of generality.
In addition, we assume that any node $v \in V_G$ has information on the values of $r^0_G(f,v)$ and $r^1_G(f,v)$ at the
beginning of the algorithm. This assumption is realized by activating $\MV(M, n, f)$ as a preprocessing step.

The key idea of our proof is to construct a \emph{sparse certificate} $H$, 
which is a spanning subgraph $H \subseteq G$ of $O(n)$ edges satisfying $H \succ G$.
If such a graph is obtained, the trivial centralized approach (i.e., the approach in which $f$ collects the whole topological information of $H$) yields an $O(n)$-round algorithm for constructing the augmenting path. For constructing sparse certificates, we first introduce a novel tree structure associated with $G$, $M$, and $f$:
\begin{definition}[Alternating base tree]
An \emph{alternating base tree} for $G$, $M$, and $f$ is the rooted spanning tree $T$ of $G$ satisfying the following conditions:
\begin{itemize}
\item $f$ is the root of $T$.
\item For any $v \in V(G)$, the edge from $v$ to its parent in $T$ is the last edge of the shortest alternating path from $f$ to $v$ in $G$. Formally, letting $\parent_T(v)$ be the parent of $v \in V(G) \setminus \{f\}$ in $T$, $r^{\gamma_G(v)}_G(f,v) = r^{1 - \gamma_G(v)}_G(f,\parent_{T}(v)) + 1$ and $\Indi_M((v, \parent_{T_I}(v))) = 1 - \gamma_G(v)$ hold for any $v \in V(G) \setminus \{f\}$.
\end{itemize}
\end{definition}
\begin{figure*}[tb]
\begin{center}
\includegraphics[width=13cm]{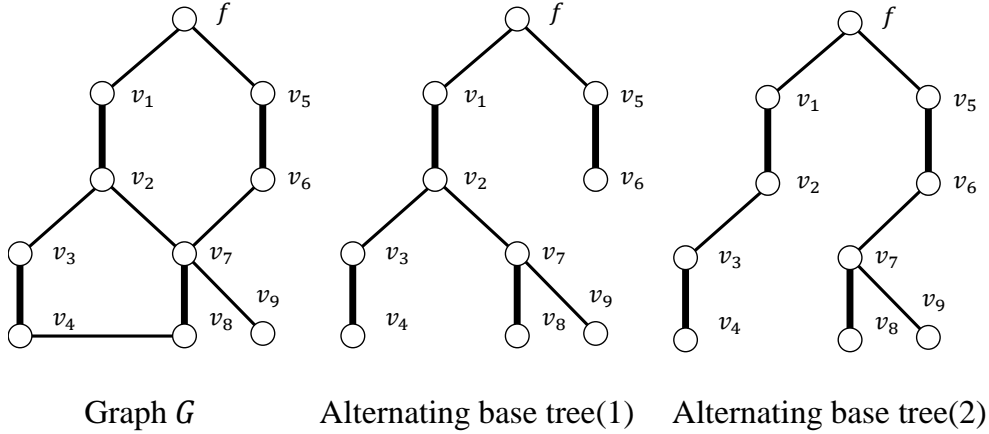}
\end{center}
\caption{Examples of the alternating base tree. Bold lines are matching edges, and thin lines are unmatched edges.}
 \label{fig:abt}
\end{figure*}
It is not difficult to check that such a spanning tree always exists. 
As a node might have two or more shortest alternating paths, $T$ is not uniquely determined (see 
Figure~\ref{fig:abt} (1) and (2) for examples). In the following argument, however, we fix an arbitrarily chosen alternating
base tree $T$. It should be emphasized that the alternating base tree does not necessarily 
contain an alternating path from $f$ to each node $v$. For example, both alternating base trees in 
Figure~\ref{fig:abt} have no alternating path from $f$ to $v_{9}$. 

Fixing $T$, the subscript $T$ of the notation $\parent_T(v)$ is omitted in the following argument. We define 
$\edgep(v)$ as the edge from $v$ to its parent and $T_v$ as the subtree of $T$ rooted by $v$. 
Any non-tree edge $e = (u, w) \in E(G) \setminus E(T)$ and the unique path from $u$ to $w$ in $T$ form a
simple cycle in $G$, which is denoted by $\cyclet(e)$. 

The sparse certificate is obtained by incrementally augmenting edges to $T$. For any $1 \leq k \leq n$, we 
define the \emph{level-$k$ edge set} $F_k$
as $F_{k}= \{(u,v) \mid (u,v) \in E(G) \setminus M \wedge \max(r^{0}_{G}(f,u), r^{0}_{G}(f,v)) = k \}
\cup \{(u,v) \mid (u,v) \in M \wedge \max(r^{1}_{G}(f,u), r^{1}_{G}(f,v)) =k \}$. We also define $F_{\leq k} = \cup_{0 \leq i \leq k} F_k$ and 
$G_k = T + F_{\leq k}$. Moreover, we define $F_0 = \emptyset$ as a sentinel.
Let $B_k$ be the set of all the bridges 
(i.e., the edge forming a cut of size one) in $G_k$.
Note that $B_h$ is a subset of $E(T)$ because $T$ is a spanning tree of $G$.
The following lemma is the key technical ingredient of our construction.
\begin{lemma}
\label{lma:3}
Let $F^c_k \subseteq F_k \setminus E(T)$ be an arbitrary subset of non-tree edges in $F_k$ satisfying $B_{k-1} \setminus B_{k} \subseteq \cup_{e \in F^c_k} E(\cyclet(e))$. Then, $(T + \cup_{1 \leq i \leq k} F^c_i) \succ G_{k}$ holds.  
In addition, the edge set $F^c = \cup_{0 \leq i \leq k} F^c_i$ contains at most $n-1$ edges.
\end{lemma}
This lemma naturally yields the following incremental construction of sparse certificates: each node $v$ identifies $k$ such that
$\edgep(v) \in B_{k-1} \setminus B_k$ holds, and if $T_v$ has an outgoing edge $e$
belonging to $F_k$, $v$ adds $e$ to $F^c_k$ (if $F_k$ contains two or more outgoing edges,
one is chosen arbitrarily).
Since $\cyclet(e)$ obliviously covers $\edgep(v)$, the constructed edge set $F^c_k$ satisfies 
the lemma. Consequently, $H = T + \cup_{1 \leq i \leq n} F^c_i \succ G_n$ is satisfied, and thus, $H$ is a sparse certificate.

Considering the distributed construction of $H$, a useful property of Lemma~\ref{lma:3} 
is that one does not have to wait for the computation
of $F^c_k$ to start the computation of $F^c_{k+1}$. As the information on $r^{\theta}_G(f,v)$ for $\theta \in 
\{0, 1\}$ is available to $v$, each node can identify the level of each incident edge. 
Thus, the construction of $F^c_k$ for all $k$ can be executed in parallel. The details of 
the distributed construction is explained in Section~\ref{sec:dis}.

\subsection{Proof Details}

Before proving Lemma~\ref{lma:3}, we prove an auxiliary lemma.

\begin{lemma}
\label{lma:1}
For any $\theta \in \{0, 1\}$ and $v \in V(G) \setminus \{f \}$ such that 
$r^{\theta}_{G}(f,v) \leq k+1$ holds, $r^{\theta}_{G_{k}}(f,v) = r^{\theta}_{G}(f,v)$ holds 
for all $k' \geq k$.
\end{lemma}

\begin{proof}
The proof is based on induction on $k$. (Basis) $k = 0$: Let $v$ be any node
satisfying $r^{\theta}_G(f,v) \leq k + 1$ for some $\theta \in \{0, 1\}$, and 
let $Q$ be the $\theta$-shortest path from $f$ to $v$ in $G$. This path is 
contained in $T$ because $v$ chooses $f$ as its parent in $T$.
(Inductive Step): As the induction hypothesis, suppose 
$r^{\theta}_{G_{k - 1}}(f,u) = r^{\theta}_G(f,u)$ holds (and also $r^{\theta}_{G_k}(f,u) = r^{\theta}_G(f,u)$ holds because of $G_{k-1} \subseteq G_k$) for any $u$ and $\theta$ satisfying $r^{\theta}_G(f,u) \leq k$.
Consider any node $v$ such that $r^{\theta}_G(f,v) \leq k+1$ holds. As 
the case of 
$r^{\theta}_{G}(f,v) < k + 1$ is evidently proved by the induction hypothesis, 
we assume $r^{\theta}_{G}(f,v) = k + 1$. The proof consists of the
following two cases.

\noindent
\textbf{(Case 1)} $\gamma_G(v) = \theta$: By the definition of alternating base trees, we have $r^{1-\theta}_{G}(f,\parent(v)) = r^{\theta}_{G}(f,v) - 1 = k$.
In addition, for any $w \in T$, $r^{\gamma_G(w)}_{G}(f,w) = r^{1-\gamma_G(w)}_{G}(f,\parent(w))+1 >  r^{\gamma_{G}(\parent(w))}_{G}(f,\parent(w))$ holds.
Therefore any node $w \in T_v$ satisfies $r^{\gamma_G(w)}_G(f,w) \geq k + 1$. Then, any outgoing non-tree edge of $T_v$ has a level of at least $k+1$. That is, $\edgep(v)$ is the bridge in $G_k$. 
Since $r^{1 - \theta}_{G}(f,\parent(v)) = k$ holds, the induction hypothesis
yields $r^{1 - \theta}_{G_k}(f,\parent(v)) = k$ and thus 
there exists a $(1 - \theta)$-alternating path $P$ from $f$ to $\parent(v)$ in $G_k$.
Due to the fact that $\edgep(e)$ is a bridge, $P$ does not contain $v$.
Hence the concatenated path $P \circ \edgep(e)$ is a $\theta$-alternating path
from $f$ to $v$ in $G_k$ of length $k + 1$. That is, 
$r^{\theta}_{G_k}(f,v) = r^{\theta}_G(f,v)$ holds.

\noindent
\textbf{(Case 2)} $\gamma_G(v) = 1 - \theta$: 
Let $Q = v_0, e_1, v_1, e_2, \dots, e_{k+1}, v_{k+1}$ be the shortest $\theta$-alternating path from $f$ to $v$ in $G$ ($f = v_0$ and $v =  v_{k+1}$).
To prove the lemma, it suffices to show that any edge in $Q$ has a level of at most $k$ or is an edge in $E(T)$. Suppose for contradiction that a non-tree 
edge $e_j$ has the level $k' > k$. Without loss of generality, we assume that 
$j$ is the highest value for which this condition is satisfied. That is, any edge $e_{j'}$ 
for $j' > j$ has a level of at most $k$ or an edge in $T$. We define $\rho$ as $\Indi_M(e_j)$. We further divide Case 2 into the following three subcases.

\noindent
\textbf{(Case 2a)} $j = k + 1$: Since $Q$ is the shortest $\theta$-alternating path of length $k+1$, $\rho = 1 - \theta$ holds, and $Q^{p}_{v_k}$ is a 
$(1 - \theta)$-alternating path from
$f$ to $v_k$ of length $k$.
From the condition $\gamma_G(v) = \gamma_G(v_{k+1}) = 1 - \theta$ for Case 2, 
$r^{1 - \theta}_G(f,v_k) \leq k$ and $r^{1 - \theta}_G(f,v_{k+1}) < r^{\theta}_G(f,v_{k+1}) = k + 1$ hold. That is,
the level of $e_{j} = e_{k+1}$ is at most $k$, which is a contradiction.

\noindent
\textbf{(Case 2b)} $j < k + 1$ and $\rho = 1$: Since the length of $Q^{p}_{v_j}$ is $j$, we have $r^{0}_{G}(f,v_{j})\leq j \leq k$.
From the induction hypothesis, $G_{k-1}$ contains a $0$-alternating path $Q'$ 
from $f$ to $v_j$. In other words, $v_j$ has a $0$-alternating path $Q'$ such that any non-tree edge
in $E(Q')$ has a level of at most $k$. The assumption of $\rho = 1$ implies that 
$Q'$ must terminate with a matching edge incident to $v_j$, i.e., the edge $e_j$. This is a contradiction because we assume that $e_j$ is not contained in $G_{k-1}$. 

\noindent
\textbf{(Case 2c)} $j < k + 1$ and $\rho = 0$: 
We denote $R = Q^{s}_{v_j}$ as shorthand.
As the length of $Q^{p}_{j}$ is $j \leq k$, from the induction hypothesis,
we have $r^{1}_{G_{k}}(f,v_{j})=r^{1}_{G}(f,v_{j})\leq j$, and thus, 
there exists the shortest $1$-alternating path $P$ from $f$ to $v_j$ in $G_{k}$.
Let $v_h \in V(R) \cap V(P)$ be the first node in $P$, which also belongs to $R$.
If $e_{h+1}$ is a matched edge, $P^{p}_{v_{h}}\circ Q^{s}_{v_{h}}$ is a 
$\theta$-alternating path in $G_{k}$ (see Figure~\ref{fig:alt2}(a)), 
the length of which is bounded by $|P^{p}_{v_{h}}\circ Q^{s}_{v_{h}}|\leq |P| + (k+1 - h) \leq j + (k+1 - j) \leq k+1$. Hence, we obtain $r^{\theta}_{G_{k}}(f,v_{k+1})\leq k+1 = r^{\theta}_{G}(f,v_{k+1})$,
which is a contradiction.
If $e_{h+1}$ is an unmatched edge, $e_{h}$ is a matched edge.
Therefore, $P^{p}_{v_{h}}\circ \overline{R^{p}_{v_h}}$ is a $0$-alternating path from $f$ to $v_{j}$ in $G_{k}$ (see Figure~\ref{fig:alt2}(b)).
Since we consider the case of $\rho =0$, the edge $e_{j}$ is an unmatched edge.
Therefore, $v_{h}\neq v_{j}$ holds, and thus $v_{h}$ is not the last node of $P$.
This implies $|P^{p}_{v_{h}}| \leq j-1$.
We obtain $|P^{p}_{v_{h}}\circ \overline{R^{p}_{v_h}}|\leq j-1 + (k+1-h) \leq j-1 + (k+1-j) \leq k$, and thus, $r^{0}_{G}(f,v_j)\leq r^{0}_{G_{k}}(f,v_j)\leq k$.
Since $Q^{s}_{v_{j-1}}$ is a $0$-alternating path from $f$ to $v_{j-1}$ of length $j-1$, we have $r^{0}_{G}(f,v_{j-1})\leq j-1 \leq k$.
This implies that the level of $e_{v_{j}}$ is at most $k$,
which is a contradiction.
\end{proof}

\begin{figure*}[tb]
\begin{center}
\includegraphics[width=13cm]{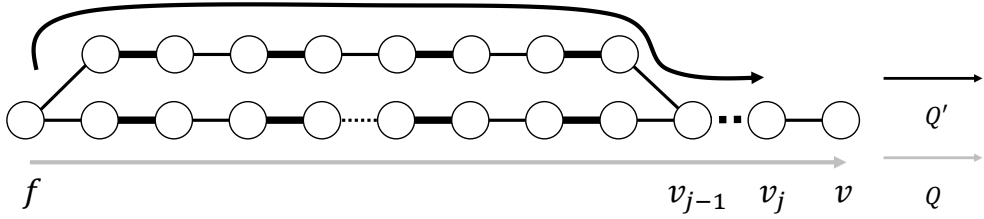}
\end{center}
 \caption{Proof of Lemma~\ref{lma:1} for (Case 2b). Bold lines are matching edges, and thin lines are unmatched edges. The dotted line is the edge included in $G$ but not in $G_{k}$.
Note that the edge $(v_{j-1},v_{j})$ is actually included in $G_{k}$, but it is drawn with a dotted line for explaining the contradiction. 
 }
 \label{fig:alt1}
\end{figure*}

\begin{figure*}[tb]
\begin{center}
\includegraphics[width=13cm]{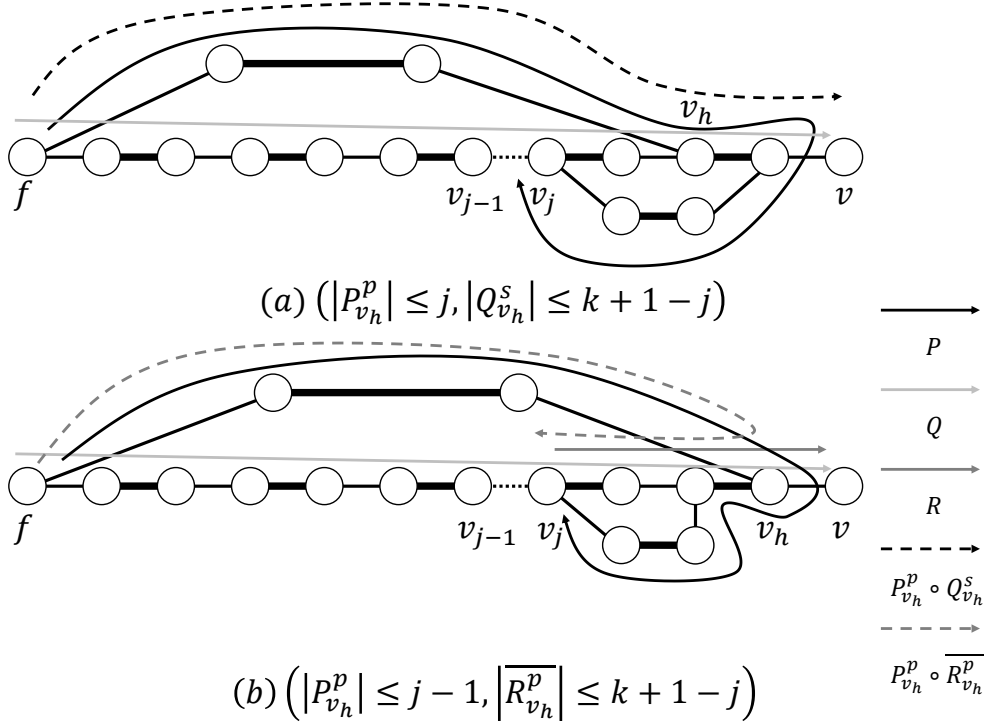}
\end{center}
 \caption{Proof of Lemma~\ref{lma:1} of (Case 2c). Bold lines are matching edges, and thin lines are unmatched edges. The dotted line is the edge included in $G$ but not in $G_{k}$.}
 \label{fig:alt2}
\end{figure*}

Now, we present the proof of Lemma~\ref{lma:3}.
\begin{proof}
 
Let $F^c_{\leq k} = \cup_{1 \leq i \leq k} F^c_i$ and $H_k = T + F^c_{\leq k}$. 
We prove the lemma inductively. For $k = 0$, $H_0 = T \succ G_0 = T$ evidently holds.
Thus, it suffices to show $H_k \succ G_k$, assuming $H_{k'} \succ G_{k'}$ for all 
$0 \leq k' < k$. For any $0 \leq h \leq n$, we define 
$U_h = \{(v, \theta) \mid v \in V(G) \wedge r^{\theta}_{G_k}(f,v) = h\}$. If
$v$ is $\theta$-reachable in $H_k$ for all $0 \leq h \leq n$ and $(v, \theta) \in U_h$, 
we can conclude that $H_k \succ G_k$. The proof of this statement follows the (nested) 
induction on $h$. (Basis) As $U_0$ contains only $(f, 1)$, the statement evidently holds. (Inductive Step) As the induction hypothesis, suppose $v$ is $\theta$-reachable for
any $(v, \theta) \in \cup_{0 \leq i \leq h} U_i$, and consider any pair $(v, \theta$) 
in $U_{h+1}$. Then, we consider the following two cases.

\noindent
\textbf{(Case 1)} $\edgep(v)$ is a bridge in $G_k$: We have 
$r^{1 -\theta}_G(f,\parent(v)) = h $ from the definition of alternating base trees. 
Since the induction hypothesis guarantees that $\parent(v)$ preserves the reachability
of $G_{k}$ in $H_k$, there exists a $(1 - \theta)$-alternating path $P$ from $f$ to $\parent(v)$ in $H_k$. In addition, $P$ does not contain $\edgep(e)$, because 
$\edgep(e)$ is a bridge in $H_k \subseteq G_k$. From
$\Indi_M(\edgep(v)) = 1 - \theta$, which directly follows from the definition of alternating base trees, the concatenated path $P \circ \edgep(v)$ becomes a $\theta$-alternating 
path from $f$ to $v$ in $H_k$ (see Figure~\ref{fig:bridge}~(1)). Then, $v$ is $\theta$-reachable in $H_k$.

\noindent
\textbf{(Case 2)} $\edgep(v)$ is not a bridge in $G_k$: As 
$G_0 \subseteq G_1 \subseteq \dots, \subseteq G_k$ holds, there exists $1 \leq j \leq k$ 
such that $\edgep(G) \in B_{j - 1} \setminus B_{j}$ holds. Then, $F^c_{j}$ contains an 
outgoing edge $e$ of $T_v$ belonging to $F_{j}$. 
Let $e = (u, w)$ and $u$ be the side contained in $T_v$. We assume that $e$ is not a matching edge. By symmetry, the case of $e \in M$ is proved similarly. 
From the definition of $F_{j}$, we have $\max\{r^{1}_G(f,u), r^{1}_G(f,w)\} = j \leq k$. Lemma~\ref{lma:1} implies that both $u$ and $w$ have $1$-alternating paths from 
$f$ in $G_{j - 1}$; from the induction hypothesis $H_{j - 1} \succ G_{j-1}$,
they have $1$-alternating paths from $f$ also in $H_{j-1}$, which we refer to as $P$ and $Q$, respectively. 
Since $\edgep(v)$ is a bridge of $G_{j-1} \supseteq H_{j-1}$, the suffix 
$P^s_v$ is a subgraph of $T_v$. 
In addition, $Q$ does not intersect $V(T_v)$, because both $f$ and $w$ are outside $T_v$. Thus, $P^s_v$ and $Q$ are mutually disjoint, and the concatenated path $Q' = Q 
\circ (w, u) \circ \overline{P^s_v}$ is simple. It is easy to check that $Q'$ is an 
alternating path from $f$ to $v$. As $Q$, $e$, and $\overline{P^s_v}$ are all 
contained in $H_{j-1} + F^c_{j} = H_j$, $P^p_v$ and $Q'$ are contained in $H_j$ (see Figure~\ref{fig:bridge}~(2)).
The alternating paths $P^s_v$ and $Q'$ have different parities because their last edges are adjacent in $P$. Hence, we conclude that $v$ is bireachable in $H_j$.

The remaining matter in the proof is to provide a bound on the size of $\cup_{0 \leq i \leq n-1} F^c_i$.
Because $G_k \subseteq G_{k+1}$ holds for any $0 \leq k \leq n-1$, we have 
$B_{k+1} \subset B_k$, which implies that $B_k \setminus B_{k+1}$ for all $k$ are mutually disjoint. Then, $\sum_{0 \leq i \leq n-1} |B_k \setminus B_{k+1}| = |B_0| = n-1$
holds. Since at most one edge is augmented for each edge in $B_k \setminus B_{k+1}$,
the size $|\cup_{0 \leq i \leq n-1} F^c_i|$ is also bounded by $n-1$. 
\end{proof}

\begin{figure*}[tb]
\begin{center}
\includegraphics[width=13cm]{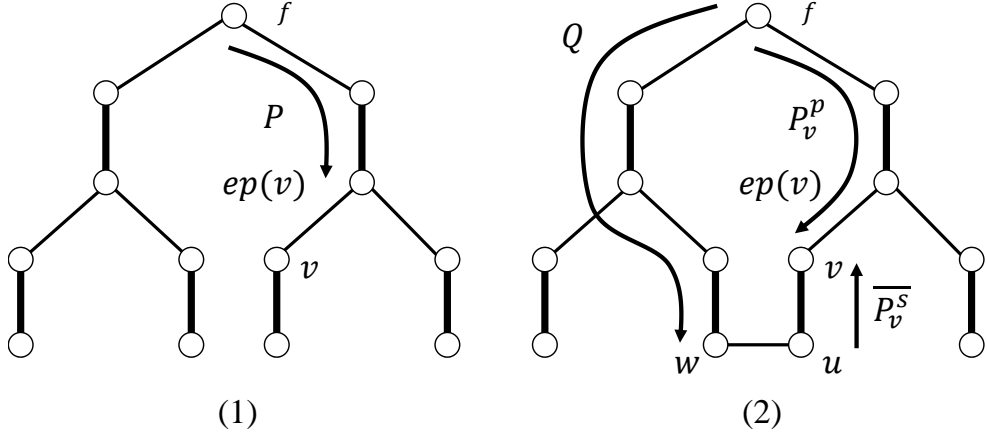}
\end{center}
 \caption{Proof of Lemma~\ref{lma:3}. Bold lines are matching edges, and thin lines are unmatched edges.}
 \label{fig:bridge}
\end{figure*}
\subsection{Distributed Implementation}
\label{sec:dis}

This section explains how to implement the centralized sparse certificate algorithm, presented in Section~\ref{sec:out_n}, in the CONGEST model to obtain the algorithm of Theorem~\ref{lma:n}.
It is relatively straightforward 
to construct the alternating base tree $T$. From the preprocessing run of $\MV(M, n, f)$, 
each node $v$ has information on the values of $r^1_G(f,v)$ and $r^0_G(f,v)$; thus, it has information on $\gamma_G(v)$ as well. 
Then, $v$ chooses an arbitrary neighbor $u$ of $v$ satisfying the second condition 
of the alternating base tree as its parent (i.e., it chooses $(v, u)$ as an edge of $T$). 
Algorithm~\ref{alg:base} presents the pseudocode of the alternative base tree construction. This algorithm is a local algorithm, which is implemented in zero round.

The main idea of constructing the edge set $F^c = \cup_{1 \leq i \leq n} F^c_i$ in the 
distributed manner is implemented by the CONGEST algorithm \textsf{ConstF$(k)$}, where each node $v$ outputs an outgoing edge of $T_v$ of level $k$ if it exists (or $\perp$ otherwise).
Let $d$ be the height of the constructed alternating base tree $T$.
Given a non-tree edge $e = (u, w) \in E(G) \setminus E(T)$, the depth of the lowest common ancestor of $u$ and $w$ is denoted by $\lca(e)$. 
In addition, we introduce the ordering relation $\leq_{\lca}$ over all 
non-tree edges as $e_1 \leq_{\lca} e_2 $ if and only if $\lca(e_1) \leq \lca(e_2)$.
The algorithm \textsf{ConstF} works under the assumption that 
for any non-tree edge $e = (u, v)$, $u$ and $v$ have information on the value of $\lca(e)$. This assumption is
realized by the following $O(d)$-round preprocessing.
\begin{enumerate}
\item Each node $v$ computes its depth $d_v$ in $T$ through a downward message propagation from $f$ along $T$. The root $f$ first sends to its children the value one. The node $v$ receiving message $i$ decides $d_v = i$ and sends the value $i+1$ to its chldren. 
\item Each node $v$ broadcasts the pair of its ID and depth $(v, d_v)$ to all the nodes in $T_v$. First, each node sends the pair to its children. In the following rounds, each node forwards the message from its parents to the children. This task finishes within $O(d)$ rounds. 
\item The broadcast information of the previous step allows each node $v$ to identify the 
path $p_T(v)$ from $v$ to $f$ in $T$. For all non-tree edges $e = (u, v)$, $u$ and $v$ exchange $p_T(v)$ (taking $O(d)$ rounds) and compute the value of $\lca(e)$.
\end{enumerate}

The pseudocode of Algorithm \textsf{ConstF$(k)$} is presented in Algorithm~\ref{alg:constf}.
Let $E^{\ast}(T_v)$ be the set of non-tree edges $e$ such that at least one endpoint 
of $e$ belongs to $V(T_v)$. Each node $v$ computes the minimum edge 
$e_v \in E^{\ast}(T_v) \cap F_k$ with respect to $\leq_{\lca}$. This task is 
implemented through a standard aggregation over $T$. Each leaf node $v$ sends 
the minimum edge $e$ in $F_k \cap E^{\ast}(T_v)$. If $F_k \cap E^{\ast}(T_v) = \emptyset$
holds, the leaf sends a dummy edge $e$ such that $\lca(e) = \infty$ holds (the edge sent to
the parent is implicitly associated with the value of $\lca(e)$ to admit the comparison based on $\leq_{\lca}$). Let $X$ be the set of edges a non-leaf node $v$ received from 
its children. Then, $v$ chooses $e_v$ as the minimum edge in 
$X \cup (I(v) \cap F_k \cap E^{\ast}(T_v))$ with respect to $\leq_{\lca}$ and sends the chosen edge to
$\parent(v)$. Finally, $v$ outputs $e_v$ if $\lca(e_v) < d_v$ holds or $\perp$ otherwise.
The correctness of \textsf{ConstF$(k)$} follows the proposition below.
\begin{proposition}
Let $e$ be the minimum edge in $E^{\ast}(T_v)$ with respect to $\leq_{\lca}$. 
Then, $e$ is an outgoing edge of $T_v$ if and only if $\lca(e) < d_v$ holds (thus,
$\edgep(v)$ is a bridge if $\lca(e) \geq d_v$ holds).
\end{proposition}

The edge set $F^c$ is constructed by running 
\textsf{ConstF$(k)$} for all $1\leq k \leq n$. As this algorithm is implemented by one-shot aggregation over $T$, one can utilize the standard pipelining technique for completing \textsf{ConstF$(k)$} for all 
$1 \leq k \leq n$, which takes $O(n)$ rounds in total (including 
the preprocessing step of computing $\lca(e)$). 
The result of \textsf{ConstF} provides node $v$ with the information of 
the minimum $k$, such that $\edgep(v) \in B_{k-1} \setminus B_k$, as well as 
an outgoing edge of $T_v$ in $F_k$.
Following Lemma~\ref{lma:1}, each node $v$ can decide the edge $e$ that should be added to $F^c = \cup_{1 \leq i \leq n} F^c_i$. 

\begin{algorithm*}[tb]
\caption{Construction of the alternating base tree for $v_{i}$: \textsf{ABT}$((G,M))$}
\label{alg:base}
\begin{algorithmic}[1]
\REQUIRE The graph induced by the edge set $\bigcup_{i:v_{i}\in V}E_{i}$ is an alternating base tree.
\STATE $E_{i}$: initially $\empty$  $\emptyset$.
\IF {$v\neq f$}
\STATE choose edge $(u,v)$ that is incident on the vertex $v$ and satisfies $r^{\gamma(v)}_{I}(f,v)=r^{1-\gamma(v)}_{I}(f,u)$ and $\Indi((u,v))=1-\gamma(v)$ (if multiple edges satisfy these conditions, the node arbitrarily chooses one).
\STATE $E_{i}\leftarrow E_{i} \cup {(u,v)}$.
\ENDIF
\end{algorithmic}
\end{algorithm*}
\begin{algorithm*}[tb]
\caption{Construction of $F^{c}_{k}$ for $v_{i}$: \textsf{ConstF}$(k)$}
\label{alg:constf}
\begin{algorithmic}[1]
\REQUIRE The edge $e_{i}$ is an outgoing edge of $T_{v_{i}}$ if node $v_{i}$ outputs $e_{i}$; otherwise, $T_{v_{i}}$ does not have an outgoing edge.
\FOR{$i=1;i\leq d;i ++$}
	\IF{$v_{i}$ is a leaf node}
		\IF{$I(v_{i})\cap F_{k}\cap E^{*}(T_{v})=\emptyset$}
			\STATE $e_{v_{i}}\leftarrow$ dummy edge $e$ such that $\lca(e)=\infty$.
		\ELSE 
			\STATE $e_{v_{i}} \leftarrow \min_{e\in I(v_{i})\cap F_{k}\cap E^{*}(T_{v_{i}})}e$ w.r.t.~$\leq_{\lca}$.
		\ENDIF
		\IF{$v_{i}\neq f$}
			\STATE send $e_{v_{i}}$ to its parent.
		\ENDIF
	\ELSE 
			\IF{$v_{i}$ receives the set of edges $X$ from all its children}
				\STATE $e_{v_{i}} \leftarrow \min_{e\in X\cup (I(v_{i})\cap F_{k}\cap E^{*}(T_{v_{i}}))}e$ w.r.t.~$\leq_{\lca}$.
			\ENDIF
	\ENDIF
\ENDFOR
\IF{$\lca(e_{v_{i}})\leq d(v_{i})$}
	\STATE output $e_{v}$.
\ELSE
	\STATE output $\perp$.
\ENDIF
\end{algorithmic}
\end{algorithm*}

\section{Conclusion}
We proposed the randomized $O(s_{\max}^{3/2}+\log n)$-rounds (i.e. $O(n^{3/2})$-rounds) algorithm for computing a maximum matching in the CONGEST model, which is the first one of attaining $o(n^2)$-round
complexity for general graphs. Our algorithm follows the standard augmenting-path approach, and the technical core lies two fast algorithms of finding augmenting paths respectively running in $O(\ell^2)$ and $O(s_{\max})$ rounds.

While we believe that our result is a big step toward the goal of revealing the tight round complexity 
of the exact maximum matching problem, the gap between the upper and lower bounds are still large. It
should be noted that we leave the possibility of much faster augmenting path algorithms. Once 
an $o(\ell^2)$-round or $o(s_{\max})$-round algorithm of finding an augmenting path is invented,
the upper bound automatically improves. This direction is still promising.
\section*{Acknowledgement}
This work was supported by JSPS KAKENHI Grant Numbers  JP19J22696, 20H04140, 20H04139,
and 19K11824.
%

%
%

%
\nocite{*}

\bibliographystyle{plain}
\bibliography{reference}

\begin{thebibliography}{10}

\bibitem{AK20}
Mohamad Ahmadi and Fabian Kuhn.
\newblock Distributed maximum matching verification in congest.
\newblock In {\em 34th International Symposium on Distributed Computing
  (DISC)}, pages 37:1--37:18, 2020.

\bibitem{AKO18}
Mohamad Ahmadi, Fabian Kuhn, and Rotem Oshman.
\newblock Distributed approximate maximum matching in the congest model.
\newblock In {\em 32rd International Symposium on Distributed Computing
  (DISC)}, pages 6:1--6:17, 2018.

\bibitem{bacrach2019hardness}
Nir Bacrach, Keren Censor-Hillel, Michal Dory, Yuval Efron, Dean Leitersdorf,
  and Ami Paz.
\newblock Hardness of distributed optimization.
\newblock In {\em 2019 ACM Symposium on Principles of Distributed Computing
  (PODC)}, pages 238--247, 2019.

\bibitem{BCGS17}
Reuven Bar-Yehuda, Keren Censor-Hillel, Mohsen Ghaffari, and Gregory
  Schwartzman.
\newblock Distributed approximation of maximum independent set and maximum
  matching.
\newblock In {\em 36th annual ACM Symposium on Principles of Distributed
  Computing (PODC)}, pages 165--174, 2017.

\bibitem{BKS18}
Ran Ben-Basat, Ken-ichi Kawarabayashi, and Gregory Schwartzman.
\newblock Parameterized distributed algorithms.
\newblock In {\em 33rd International Symposium on Distributed Computing
  (DISC)}, pages 6:1--6:16, 2018.

\bibitem{BN18}
Aaron Bernstein and Danupon Nanongkai.
\newblock Distributed exact weighted all-pairs shortest paths in near-linear
  time.
\newblock In {\em Proc. of the 51st Annual ACM SIGACT Symposium on Theory of
  Computing (STOC)}, page 334–342, 2019.

\bibitem{Blum1990}
Norbert Blum.
\newblock A new approach to maximum matching in general graphs.
\newblock In {\em International Colloquium on Automata, Languages, and
  Programming}, pages 586--597, 1990.

\bibitem{KS17}
Keren Censor-Hillel, Seri Khoury, and Ami Paz.
\newblock {Quadratic and Near-Quadratic Lower Bounds for the CONGEST Model}.
\newblock In {\em 31st International Symposium on Distributed Computing
  (DISC)}, pages 10:1--10:16, 2017.

\bibitem{CPSZ19}
Yi-Jun Chang, Seth Pettie, and Hengjie Zhang.
\newblock Distributed triangle detection via expander decomposition.
\newblock In {\em Thirtieth Annual ACM-SIAM Symposium on Discrete Algorithms},
  pages 821--840, 2019.

\bibitem{daga2019distributed}
Mohit Daga, Monika Henzinger, Danupon Nanongkai, and Thatchaphol Saranurak.
\newblock Distributed edge connectivity in sublinear time.
\newblock {\em arXiv preprint arXiv:1904.04341}, 2019.

\bibitem{DEMN20}
Michal Dory, Yuval Efron, Sagnik Mukhopadhyay, and Danupon Nanongkai.
\newblock Distributed weighted min-cut in nearly-optimal time.
\newblock {\em arXiv}, 2020.

\bibitem{Edmonds2}
Jack Edmonds.
\newblock Maximum matching and a polyhedron with 0,1-vertices.
\newblock {\em Journal of Research of the National Bureau of Standards Section
  B Mathematics and Mathematical Physics}, page 125, 1965.

\bibitem{Edmonds}
Jack Edmonds.
\newblock Paths, trees, and flowers.
\newblock {\em Canadian Journal of mathematics}, pages 449--467, 1965.

\bibitem{FGKO18}
Orr Fischer, Tzlil Gonen, Fabian Kuhn, and Rotem Oshman.
\newblock Possibilities and impossibilities for distributed subgraph detection.
\newblock In {\em 30th on Symposium on Parallelism in Algorithms and
  Architectures (SPAA)}, pages 153--162, 2018.

\bibitem{FN18}
Sebastian Forster and Danupon Nanongkai.
\newblock A faster distributed single-source shortest paths algorithm.
\newblock In {\em 59th {IEEE} Annual Symposium on Foundations of Computer
  Science {(FOCS)}}, pages 686--697, 2018.

\bibitem{FHW12}
Silvio Frischknecht, Stephan Holzer, and Roger Wattenhofer.
\newblock Networks cannot compute their diameter in sublinear time.
\newblock In {\em Proc. of the Twenty-Third Annual ACM-SIAM Symposium on
  Discrete Algorithms (SODA)}, pages 1150--1162, 2012.

\bibitem{GT91}
Harold~N Gabow and Robert~E Tarjan.
\newblock Faster scaling algorithms for general graph matching problems.
\newblock {\em Journal of the ACM (JACM)}, pages 815--853, 1991.

\bibitem{GH16}
Mohsen Ghaffari and Bernhard Haeupler.
\newblock Distributed algorithms for planar networks {II:} low-congestion
  shortcuts, mst, and min-cut.
\newblock In {\em Proceedings of the twenty-seventh annual ACM-SIAM symposium
  on Discrete algorithms (SODA)}, pages 202--219, 2016.

\bibitem{GF15}
Mohsen Ghaffari, Andreas Karrenbauer, Fabian Kuhn, Christoph Lenzen, and Boaz
  Patt-Shamir.
\newblock Near-optimal distributed maximum flow.
\newblock In {\em 2015 ACM Symposium on Principles of Distributed Computing
  (PODC)}, pages 81--90, 2015.

\bibitem{ghaffari2013distributed}
Mohsen Ghaffari and Fabian Kuhn.
\newblock Distributed minimum cut approximation.
\newblock In {\em International Symposium on Distributed Computing}, pages
  1--15. Springer, 2013.

\bibitem{GKM17}
Mohsen Ghaffari, Fabian Kuhn, and Yannic Maus.
\newblock On the complexity of local distributed graph problems.
\newblock In {\em 49th Annual ACM SIGACT Symposium on Theory of Computing
  (STOC)}, pages 784--797, 2017.

\bibitem{GL18}
Mohsen Ghaffari and Jason Li.
\newblock Improved distributed algorithms for exact shortest paths.
\newblock In {\em Proc. of the 50th Annual ACM SIGACT Symposium on Theory of
  Computing (STOC)}, pages 431--444, 2018.

\bibitem{GL18-2}
Mohsen Ghaffari and Jason Li.
\newblock New distributed algorithms in almost mixing time via transformations
  from parallel algorithms.
\newblock In {\em Proceedings of 32nd International Symposium on Distributed
  Computing (DISC)}, pages 31:1--31:16, 2018.

\bibitem{grandoni2008distributed}
Fabrizio Grandoni, Jochen K{\"o}nemann, and Alessandro Panconesi.
\newblock Distributed weighted vertex cover via maximal matchings.
\newblock {\em ACM Transactions on Algorithms (TALG)}, pages 1--12, 2008.

\bibitem{HIZ16-2}
Bernhard Haeupler, Taisuke Izumi, and Goran Zuzic.
\newblock Near-optimal low-congestion shortcuts on bounded parameter graphs.
\newblock In {\em Proceedings of 30nd International Symposium on Distributed
  Computing (DISC)}, pages 158--172, 2016.

\bibitem{H18}
Bernhard Haeupler and Jason Li.
\newblock Faster distributed shortest path approximations via shortcuts.
\newblock In {\em 32nd International Symposium on Distributed Computing
  (DISC)}, pages 33:1--33:14, 2018.

\bibitem{HW12}
Stephan Holzer and Roger Wattenhofer.
\newblock Optimal distributed all pairs shortest paths and applications.
\newblock In {\em Proc. of the 2012 ACM Symposium on Principles of Distributed
  Computing (PODC)}, pages 355--364, 2012.

\bibitem{HK73}
John~E Hopcroft and Richard~M Karp.
\newblock An n\^{}5/2 algorithm for maximum matchings in bipartite graphs.
\newblock {\em SIAM Journal on computing}, pages 225--231, 1973.

\bibitem{II86}
Amos Israeli and Alon Itai.
\newblock A fast and simple randomized parallel algorithm for maximal matching.
\newblock {\em Information Processing Letters}, pages 77--80, 1986.

\bibitem{jurdzinski2018mst}
Tomasz Jurdzinski and Krzysztof Nowicki.
\newblock {MST} in \emph{O}(1) rounds of congested clique.
\newblock In {\em Proceedings of the Twenty-Ninth Annual ACM-SIAM Symposium on
  Discrete Algorithms (SODA)}, pages 2620--2632, 2018.

\bibitem{KKOI2019}
Naoki Kitamura, Hirotaka Kitagawa, Yota Otachi, and Taisuke Izumi.
\newblock Low-congestion shortcut and graph parameters.
\newblock In {\em Proccedings of 33rd International Symposium on Distributed
  Computing (DISC)}, pages 25:1--25:17, 2019.

\bibitem{koufogiannakis2009distributed}
Christos Koufogiannakis and Neal~E Young.
\newblock Distributed and parallel algorithms for weighted vertex cover and
  other covering problems.
\newblock In {\em 28th ACM symposium on Principles of distributed computing
  (PODC)}, pages 171--179, 2009.

\bibitem{FTR06}
Fabian Kuhn, Thomas Moscibroda, and Roger Wattenhofer.
\newblock The price of being near-sighted.
\newblock In {\em 17th Annual ACM-SIAM Symposium on Discrete Algorithms
  (SODA)}, pages 1109557--1109666, 2006.

\bibitem{KMW16}
Fabian Kuhn, Thomas Moscibroda, and Roger Wattenhofer.
\newblock Local computation: Lower and upper bounds.
\newblock {\em Journal of the ACM (JACM)}, pages 1--44, 2016.

\bibitem{SD98}
Shay Kutten and David Peleg.
\newblock Fast distributed construction of small \emph{k}-dominating sets and
  applications.
\newblock {\em Journal of Algorithms}, pages 40--66, 1998.

\bibitem{LP13-2}
Christoph Lenzen and David Peleg.
\newblock Efficient distributed source detection with limited bandwidth.
\newblock In {\em Proc. of the 2013 ACM Symposium on Principles of Distributed
  Computing (PODC)}, pages 375--382, 2013.

\bibitem{LPP08}
Zvi Lotker, Boaz Patt-Shamir, and Seth Pettie.
\newblock Improved distributed approximate matching.
\newblock pages 1--17, 2015.

\bibitem{Nanongkai14}
Danupon Nanongkai.
\newblock Distributed approximation algorithms for weighted shortest paths.
\newblock In {\em Proc. of the 46th Annual ACM Symposium on Theory of Computing
  (STOC)}, pages 565--573, 2014.

\bibitem{nanongkai2014almost}
Danupon Nanongkai and Hsin-Hao Su.
\newblock Almost-tight distributed minimum cut algorithms.
\newblock In {\em International Symposium on Distributed Computing}, pages
  439--453. Springer, 2014.

\bibitem{PT11}
David Pritchard and Ramakrishna Thurimella.
\newblock Fast computation of small cuts via cycle space sampling.
\newblock {\em ACM Transactions on Algorithms (TALG)}, pages 1--30, 2011.

\bibitem{lowerbound}
Atish~Das Sarma, Stephan Holzer, Liah Kor, Amos Korman, Danupon Nanongkai,
  Gopal Pandurangan, David Peleg, and Roger Wattenhofer.
\newblock Distributed verification and hardness of distributed approximation.
\newblock In {\em Proceedings of the 43th Annual ACM SIGACT Symposium on Theory
  of Computing (STOC)}, pages 363--372, 2011.

\bibitem{Vazirani2020}
Vijay~V Vazirani.
\newblock A proof of the mv matching algorithm.
\newblock {\em arXiv}, 2020.

\end{thebibliography}
\end{document}